\documentclass[letterpaper, 10 pt, conference]{ieeeconf}
\IEEEoverridecommandlockouts
\usepackage{cite}
\usepackage{amsmath,amssymb,amsfonts}
\usepackage{graphicx}
\usepackage{textcomp}
\usepackage[dvipsnames]{xcolor}
\def\BibTeX{{\rm B\kern-.05em{\sc i\kern-.025em b}\kern-.08em
    T\kern-.1667em\lower.7ex\hbox{E}\kern-.125emX}}

\usepackage{graphicx}
\usepackage{hhline}
\usepackage{caption,subcaption}
\usepackage{url}
\usepackage{enumerate}

\usepackage{algorithm}
\usepackage{algpseudocode}
\usepackage{algorithmicx}

\newcommand{\elis}[1]{\normalsize{\color{red}(Elis:\ #1)}}

\newcommand{\RR}{\mathbb{R}}

\newtheorem{definition}{Definition}

\newtheorem{lemma}{Lemma}

\newtheorem{remark}{Remark}

\newtheorem{proposition}{Proposition}

\newtheorem{convention}{Convention}

\DeclareMathOperator*{\argmin}{arg\,min}

\def\longversion{1} 

\usepackage{tikz}
\usetikzlibrary{shapes.geometric} 
\usetikzlibrary{calc}
\usetikzlibrary{fit}
\usetikzlibrary{arrows,automata}

\tikzstyle{process} = [rectangle, minimum width=3cm, minimum height=1cm, text centered, text width=3cm, draw=black]
\tikzstyle{decision} = [diamond, minimum width=3cm, minimum height=1cm, text badly centered, draw=black]
\tikzstyle{cloud} = [draw, ellipse, node distance=3cm, minimum height=2em]
\tikzstyle{startstop} = [rectangle, rounded corners, minimum width=1cm, minimum height=1cm,text centered, draw=black]
\tikzstyle{arrow} = [thick,->,>=stealth]

\tikzstyle{block} = [draw, rectangle, 
    minimum height=3em, minimum width=2.9cm, text centered]
\tikzstyle{sum} = [draw, fill=blue!20, circle, node distance=1cm]
\tikzstyle{input} = [coordinate]
\tikzstyle{output} = [coordinate]
\tikzstyle{pinstyle} = [pin edge={to-,thin,black}]

\begin{document}

\title{\LARGE \bf
Efficient and Reconfigurable Optimal Planning in Large-Scale Systems Using Hierarchical Finite State Machines
}


\author{Elis Stefansson$^{1}$ and Karl H. Johansson$^{1}$
\thanks{$^{1}$School of Electrical Engineering and Computer Science, KTH Royal Institute of Technology, Sweden. Email:
        {\tt\small \{elisst,kallej\}@kth.se}. The authors are also affiliated with Digital Futures.}%
\thanks{This work was partially funded by the Swedish Foundation for Strategic Research, the Swedish Research Council, and the Knut and Alice Wallenberg~foundation.}
}

\maketitle

\maketitle

\begin{abstract}
In this paper, we consider a planning problem for a large-scale system modelled as a hierarchical finite state machine (HFSM) and develop a control algorithm for computing optimal plans between any two states. The control algorithm consists of two steps: a preprocessing step computing optimal exit costs for each machine in the HFSM, with time complexity scaling linearly with the number of machines, and a query step that rapidly computes optimal plans, truncating irrelevant parts of the HFSM using the optimal exit costs, with time complexity scaling near-linearly with the depth of the HFSM. The control algorithm is reconfigurable in the sense that a change in the HFSM is efficiently handled, updating only needed parts in the preprocessing step to account for the change, with time complexity linear in the depth of the HFSM. We validate our algorithm on a robotic application, comparing it with Dijkstra's algorithm and Contraction Hierarchies. Our algorithm outperforms both.
\end{abstract}


\section{Introduction}

\subsection{Motivation}
In today's smart societies, with increased integrability and connectivity, there is an increasing need for efficient and optimal control of large-scale systems. Here, control algorithms should not only compute optimal plans for such systems efficiently but also be reconfigurable, i.e., be able to handle changes in parts of the system with ease. 

A common approach is to consider modular systems, i.e., systems that can be decomposed into independent entities (modules). For such systems, the idea is to construct control algorithms that can separate the computation into these modules and then combine the results from the modules to obtain an optimal plan for the whole system. Moreover, a change in a modular system typically affects only some of the modules, enabling control algorithms to be reconfigurable by only updating the corresponding affected parts in the~algorithm.


In this paper, we consider large-scale systems formalised by hierarchical finite state machines (HFSMs) \cite{harel1987statecharts}. An HFSM is a machine composed of several finite state machines (FSMs) nested into a hierarchical structure. This structure naturally makes an HFSM modular. The key research question is how to exploit the modularity in the HFSM to construct optimal control algorithms that are efficient and~reconfigurable.

\subsection{An Illustrative Example}\label{an_illustrative_example}
Consider an example where a robot is moving between lab houses, schematically depicted in Fig. \ref{fig:robot_example_detailed_overview}. In this environment, the robot is assigned to scan a tube, located in one of the houses, with minimal operational cost. Therefore, the robot must compute a plan that with minimal cost makes it reach the tube and scan it. 

Here, a key insight is that the environment can be naturally modelled as an HFSM with three layers corresponding to the layers in Fig. \ref{fig:robot_example_detailed_overview}, where Layer 1 prescribes how the robot can move between the houses, Layer 2 how the robot moves inside a house, and Layer 3 how the robot can operate the lab desk and scan tubes at a given location inside a house (all formally modelled by FSMs). The first key question is then \emph{how to construct a control algorithm that computes an optimal plan efficiently, by exploiting the hierarchical structure of the~HFSM.}

\begin{figure}
	\centering
  \includegraphics[width=0.4\textwidth]{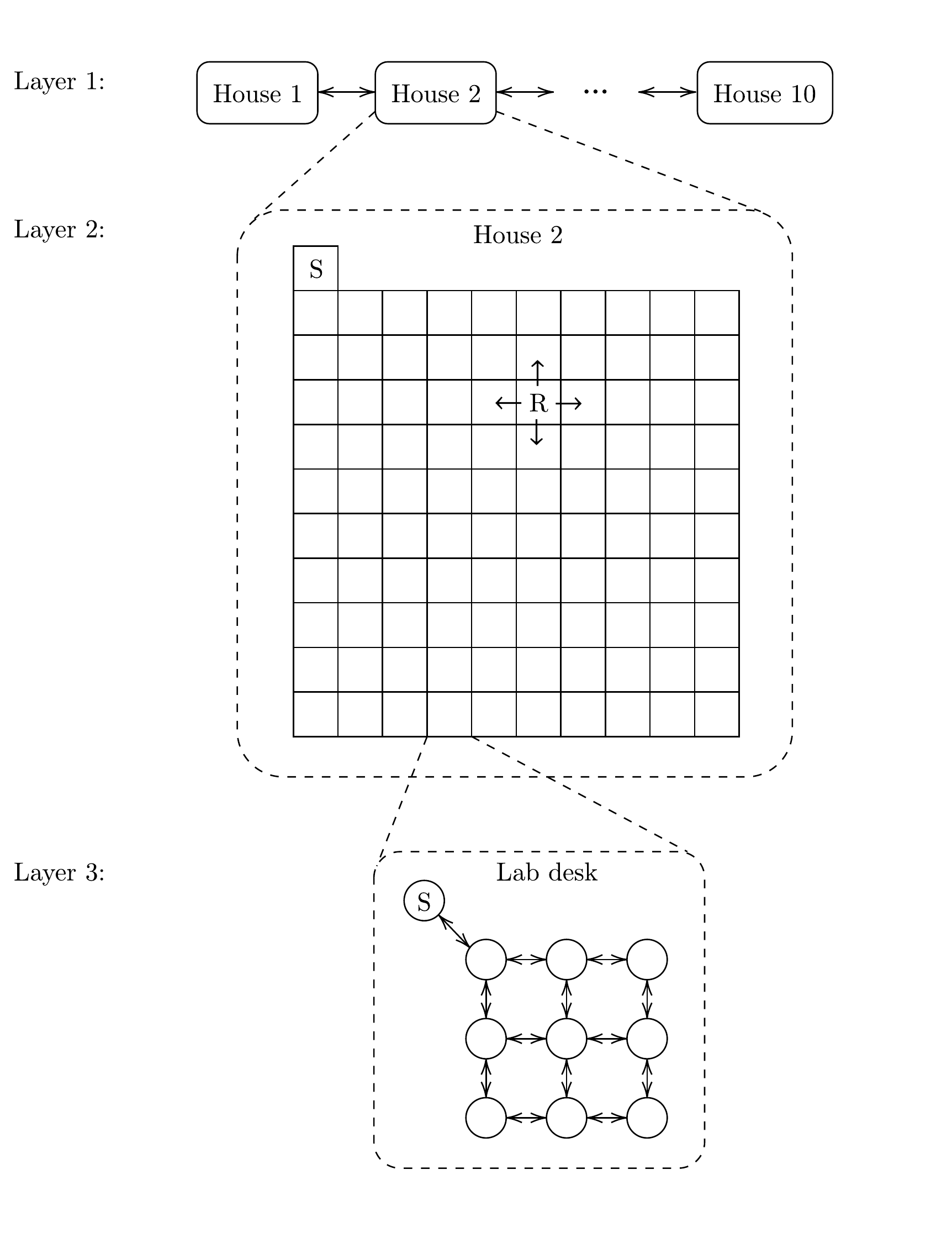}
  \caption{Robot application example modelled as an HFSM with three layers.} 
  \label{fig:robot_example_detailed_overview}
\end{figure}

Suppose now that one of the houses changes, for example, some locations in House 2 are blocked due to maintenance work, as depicted by grey areas in Fig. \ref{fig:robot_example_change_overview}. This changes the HFSM and may thus affect the optimal plan computed by the robot (e.g., if the computed optimal plan was to go to the bottom-right location in House 2, then the robot must take a detour now to avoid the blocked locations). However, only a part of the HFSM has been changed and hence, the control algorithm of the robot should not need to reset all its computations, rather, it should only need to update the relevant parts. That is, the control algorithm should be reconfigurable. The second key question is then \emph{how to construct a control algorithm that is reconfigurable given changes in the HFSM.}

\begin{figure}
	\centering
  \includegraphics[width=0.49\textwidth]{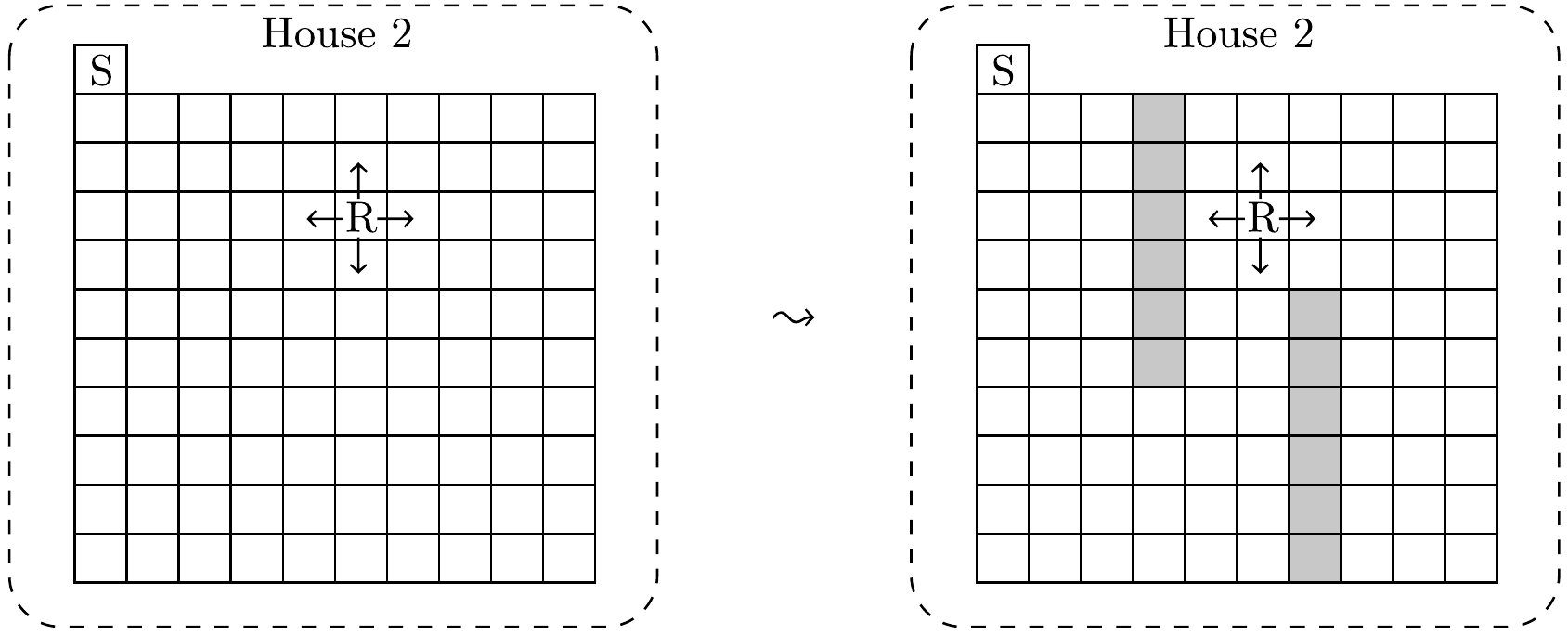}
  \caption{Change in House 2 in the HiMM given by Fig. \ref{fig:robot_example_detailed_overview}. Grey areas depict blocked locations.} 
  \label{fig:robot_example_change_overview}
\end{figure}

\subsection{Contribution}
The main contribution of this paper is to formalise optimal planning in systems modelled as HFSMs, introduce changes in such systems, and provide an efficient and reconfigurable control algorithm for computing optimal plans. More precisely, our contribution is~three-fold:

First, we formalise the HFSM systems extending the setup of \cite{biggar2021modular} to machines with outputs, i.e., Mealy machines (MMs) \cite{Mealy1955}, calling the resulting HFSM for a Hierarchical Mealy Machine (HiMM), following the formulation in \cite{stefansson2023ecc}. We then formalise changes in HiMMs, called modifications, which can change any HiMM into any other HiMM (by a series of modifications), and specify what we mean for an algorithm to be reconfigurable with respect to modifications. 

Second, we propose a control algorithm that is both efficient and reconfigurable, extending the  algorithm in \cite{stefansson2023ecc} to be reconfigurable. The algorithm is divided into two steps: A preprocessing step where optimal exit costs are computed for all MMs in the given HiMM $Z$ (done only once for a fixed HiMM), with time complexity $O(N)$, where $N$ is the number of MMs in $Z$; and a query step, rapidly computing an optimal plan between any two states in $Z$ using the optimal exit costs, with time complexity $O(\mathrm{depth}(Z) \log ( \mathrm{depth}(Z)))$ for computing the next optimal input (compared to Dijkstra's algorithm \cite{Dijkstra1959, DijkstraFibonacci} with time complexity more than exponential in $\mathrm{depth}(Z)$ in general) and is therefore regarded as efficient. Here, $\mathrm{depth}(Z)$ is the number of layers in the HiMM $Z$. The algorithm is reconfigurable in the sense that each modification of $Z$ only needs an update of the preprocessing step in time $O(\mathrm{depth}(Z))$, instead of recomputing the whole preprocessing step. In all these time complexities, we have, for brevity, assumed bounds on the number of states and inputs in an MM of $Z$, see Sections \ref{control_algorithm_theory} and \ref{more_control_algorithm_theory} for the general expressions without this~assumption.


Third, we validate our control algorithm on the robot application given in the motivation, comparing it with Dijkstra's algorithm \cite{Dijkstra1959} and Contraction Hierarchies \cite{geisberger2012exact}. We show that our control algorithm outperforms both.

\subsection{Related Work}
HFSMs are typically used to model the control law of an agent \cite{harel1987statecharts,schillinger2016human,millington2018artificial}, prescribing how the agent reacts to external inputs (e.g., ``low battery''), where FSMs in the HFSMs commonly corresponds to subtasks (e.g., ``charge battery''). Here, we instead use an HFSM formalism to model a control system where the agent can choose the inputs, for example, the robot can choose to go left in House 2 in the robot application in the motivation. The aim is to steer the system (e.g., the robot) to a desirable state with minimal cost. In discrete-event systems \cite{cassandras2008introduction}, another HFSM formalism, called state tree structures, has been developed to compute safe executions \cite{ma2006nonblocking,wang2020real}. We use a different formalism and focus instead on optimal planning minimising cost. 

There is a wide range of algorithms for optimal path planning on graphs such as Dijkstra's algorithm \cite{Dijkstra1959}, see \cite{bast2016route} for a survey. Here, a common approach is to use a preprocessing step to achieve faster planning at run-time \cite{bast2016route}, similar to our preprocessing step. In particular, hierarchical algorithms \cite{geisberger2012exact,dibbelt2016customizable,mohring2007partitioning,10.1145/1498698.1564502} using a preprocessing step are the ones most reminiscent of our approach. However, these approaches are typically based on heuristics which in the worst case perform no better than Dijkstra's algorithm. Furthermore, to use those methods, one also needs to flatten the HFSM to an equivalent FSM. This hides the hierarchical structure making these algorithms less suitable for reconfiguration. In this paper, we instead exploit the hierarchical structure to give formal performance guarantees and make the algorithm reconfigurable to changes. The paper \cite{timo2014reachability} considers a variant of an HFSM formalism and seeks an execution with the minimal length between two configurations. In this paper, we have non-negative costs, where minimal length is a special case by setting all costs to unit costs. 

In \cite{biggar2021modular}, HFSMs without outputs are formalised, and used to decompose an FSM into an HFSM. However, planning was not considered. In this paper, we extend their formalism to HFSMs with outputs and consider optimal planning.

Finally, this paper builds on the work \cite{stefansson2023ecc}. We use the same HFSM formalism as in \cite{stefansson2023ecc} and extend it to formally define changes in such a system, called modifications. We then extend the control algorithm in \cite{stefansson2023ecc} to be reconfigurable, able to handle modifications efficiently. 


\subsection{Outline}
The outline of the remaining paper is as follows. Section~\ref{problem_formulation} formalises HiMMs, including modifications of HiMMs, and formulates the problem statement. Section \ref{control_algorithm_theory} presents the efficient control algorithm, and Section \ref{more_control_algorithm_theory} extends it to be reconfigurable. Section \ref{numerical_evaluations} then conducts numerical evaluations. Finally, Section \ref{conclusion} concludes the paper with a discussion and future directions. \if\longversion0 An extended version of this paper can be found at \cite{stefansson2023cdc} that contains all the proofs in the Appendix.\else
All proofs are in the Appendix.
\fi

\subsection{Notation}
Let $f: A \rightharpoonup B$ denote a partial function from set $A$ to set $B$, and let $f(a) = \emptyset$ mean that $f$ is not defined for $a \in A$. For a tree $T$, we use the notation $(X \xrightarrow{v} Y) \in T$ meaning that there is an arc from $X$ to $Y$ labelled $v$ in $T$, and denote the depth of $T$ by $\mathrm{depth}(T)$, that is, the maximum number of arcs one can traverse in $T$ starting from the root. Finally, $|A|$ is the cardinality of the set $A$, and $\RR^+ = \{ x \geq 0: x \in \RR\}$.

\section{Problem Formulation}\label{problem_formulation}

\subsection{Hierarchical Mealy Machines}

In this section, we formalise the notion of an HiMM, also found in \cite{stefansson2023ecc}, extending  the formalism in \cite{biggar2021modular} to machines with outputs. We start with the notion of an MM.

\begin{definition}[Mealy Machine]
An MM is a tuple $M = (Q,\Sigma,\Lambda,\delta,\gamma,s)$, where $Q$ is the finite set of states, $\Sigma$ the finite set of inputs, $\Lambda$ the finite set of outputs, $\delta: Q \times \Sigma \rightharpoonup Q$ the transition function, $\gamma: Q \times \Sigma \rightarrow \Lambda$ the output function, and $s \in Q$ the start~state. We sometimes use the notation $Q(M)$, $\Sigma(M)$, $\Lambda(M)$, $\delta_M$, $\gamma_M$ and $s(M)$ to stress that e.g., $Q(M)$ is the set of states of $M$.
\end{definition}
The intuition is as follows. The MM $M = (Q,\Sigma,\Lambda,\delta,\gamma,s)$ starts in $s \in Q$, and transits to the next state $q' = \delta(q,x)$ and outputs $\gamma(q,x)$ given current state $q \in Q$ and input $x \in \Sigma$, or stops if $\delta(q,x) = \emptyset$. Repeating this process, we obtain a trajectory of $M$:
\begin{definition}
A sequence $(q_i,x_i)_{i=1}^N \in (Q \times \Sigma)^{N}$ is a trajectory of an MM $M = (Q,\Sigma,\Lambda,\delta,\gamma,s)$ if $q_{i+1} = \delta(q_i,x_i) \in Q$ for $i \in \{1,\dots,N-1\}$.
\end{definition}
In our setting, we assume that an agent (e.g., a robot) can choose the inputs. Hence, we also define a plan to be a sequence $u = (x_i)_{i=1}^N \in \Sigma^N$ of inputs, and $z = (q_i,x_i)_{i=1}^N$ to be the induced trajectory of $u$ starting from $q_1$ if $z$ is a~trajectory. We next provide the definition of an HiMM.


\begin{definition}[Hierarchical Mealy Machine]\label{HiMM_def}
An HiMM is a tuple $Z = (X,T)$ consisting of a set $X$ of MMs (the machines in $Z$) with a common input set $\Sigma$ and output set $\Lambda$, and a tree $T$ with nodes given by $X$ (how the MMs are composed in $Z$). More precisely, each node $M \in X$ in $T$ has $|Q(M)|$ outgoing arcs $\{M \xrightarrow{q} N_q \}_{q \in Q(M)}$ where either $N_q \in X$ (meaning that state $q$ in $M$ corresponds to the MM $N_q$ one layer below in the hierarchy of $Z$) or $N_q = \emptyset$ (meaning that state $q$ is simply just a state having no further refinement). 
For brevity, we call $Q_Z := \cup_{X_i \in X} Q(X_i)$ the nodes of $Z$, and $S_Z := Q_Z \cap \{q: N_q = \emptyset \}$ the states of $Z$. We assume that all state sets are disjoint, that is, $Q(X_i) \cap Q(X_j) = \emptyset$ for all $X_i,X_j \in X$ such that $X_i \neq X_j$. Furthermore, we have corresponding notions for the start state, transition function, and output function:
\begin{enumerate}[(i)]
\item The start function: The start state is a function ${\mathrm{start}: X \rightarrow S_Z}$ given by 
\begin{equation*}
\mathrm{start}(X_i) =
\begin{cases}
\mathrm{start}(X_j), &  \textrm{$(X_i \xrightarrow{s(X_i)} X_j) \in T$} \\
s(X_i), & \textrm{otherwise.}
\end{cases}
\end{equation*}
\item Hierarchical transition function: Let $q \in Q(X_j)$ with $X_j \in X$ and $v=\delta_{X_j}(q,x)$. The hierarchical transition function $\psi: Q_Z \times \Sigma \rightharpoonup S_Z$ is given by
\begin{equation*}
\psi(q,x) =
\begin{cases} 
\mathrm{start}(Y), \hspace{-2pt} & v \neq \emptyset, (X_j \xrightarrow{v} Y) \in T, Y \in X \\ 
v, & v \neq \emptyset, \mathrm{otherwise} \\
\psi(w,x), & v = \emptyset, (W \xrightarrow{w} X_j) \in T, W \in X \\ 
\emptyset, & v = \emptyset, \mathrm{otherwise.}
\end{cases}
\end{equation*}
\item Hierarchical output function: Let $q \in Q(X_j)$ with $X_j \in X$ and $v=\delta_{X_j}(q,x)$. The hierarchical output function $\chi: Q_Z \times \Sigma \rightharpoonup \Lambda$ is given by
\begin{equation*}
\chi(q,x) =
\begin{cases} 
\gamma_{X_j}(q,x), \hspace{-2pt} & v \neq \emptyset \\
\chi(w,x), & v = \emptyset, (W \xrightarrow{w} X_j) \in T, W \in X \\ 
\emptyset, & v = \emptyset, \mathrm{otherwise.}
\end{cases}
\end{equation*} 
\end{enumerate}
\end{definition}

An HiMM $Z=(X,T)$ works analogously as an MM: The HiMM $Z=(X, T)$ starts in state $\mathrm{start}(M_0)$ with $M_0$ being the root MM of $T$, and for each state $q \in S_Z$ and input $x \in \Sigma$ the machine transits to $q' = \psi(q,x) \in S_Z$ with output $\chi(q,x) \in \Lambda$, or stops if $\psi(q,x) = \emptyset$. A trajectory of an HiMM is defined analogously to a trajectory of an MM.

\begin{figure}
	\centering
  \includegraphics[width=0.20\textwidth]{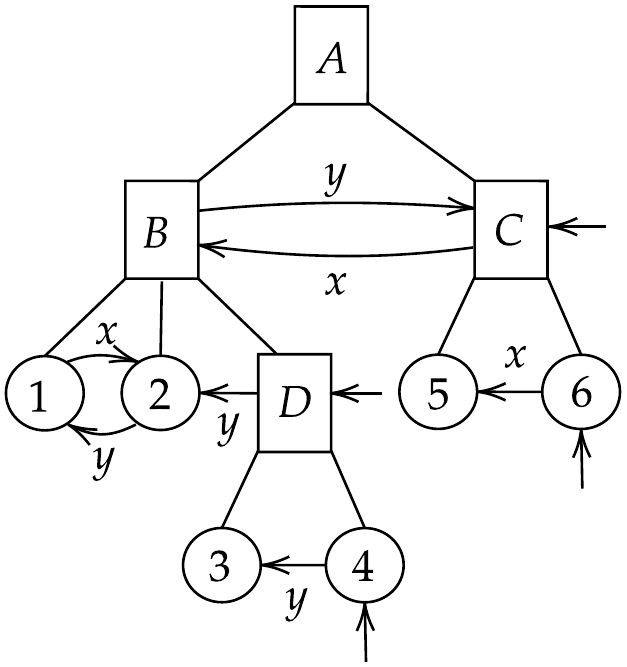} 
  \caption{Example of an HiMM depicted in a tree-form.} 
  \label{fig:Himm_transition_example_v2}
\end{figure}

\begin{remark}[Intuition] Definition \ref{HiMM_def} is technical, but the intuition is simple. To see it, depict an HiMM $Z = (X,T)$ as a tree similar to $T$ with some extra details, as in Fig. \ref{fig:Himm_transition_example_v2}. Here, $X = \{A,B,C,D\}$ are the MMs of $Z$ with dependency given by lines, e.g., $B$ is a child of $A$ in $T$. States of $Z$ are depicted as circles. We also depict the transitions of each MM given by arrows, where the label of the arrow specifies the input needed to make the transition. We omit writing the outputs explicitly. For example, in MM $A$, one goes from (the state corresponding to) $B$ to (the state corresponding to) $C$ with input $y$. Also, small arrows indicate start states (e.g., $s(D)=4$).
We now provide the intuition of Definition \ref{HiMM_def} with an example. Assume we are at state 5 and apply input $x$. If there were a transition from 5 with $x$ (as from state 6), $Z$ would just follow that transition. However, since this is not the case, $Z$ instead goes upwards iteratively in the tree $T$ until it finds a node that supports this input (or stops if it does not find one), in this case, $C$ in the MM $A$ has a transition with $x$ (to $B$). Thus, $Z$ moves to $B$. After this, $Z$ follows the start states downwards until it reaches a state of $Z$, in this case, it goes to $D$ then 4. Here, the transition is complete and thus $\psi(5,x) = 4$. Furthermore, the output $\chi(5,x)$ corresponds to the output where the transition step occurred, in our case going from $C$ with input $x$. This is to support modularity neatly where we care about that we got from $C$ with input $x$ (e.g., $x$ can encode that we did a task in $C$), and not how we managed to exit $C$ (e.g., how we did the task in $C$). Other setups are considered as future~work.

\end{remark}

\subsubsection{Optimal Planning}

In this paper, the output $\chi(q,x)$ for an HiMM $Z$ will exclusively model the cost for applying input $x$ from state $q$, therefore, we henceforth assume $\Lambda \subset \RR^+$. We also consider the cost of a trajectory $z= (q_i,x_i)_{i=1}^N$, called the cumulative cost, defined by $C(z) := \sum_{i=1}^N \chi(q_i,x_i)$ if all $\chi(q_i,x_i) \neq \emptyset$, and $C(z) = + \infty$ otherwise.\footnote{In fact, only $\chi(q_i,x_i)$ may be empty, since $z$ is a trajectory. Also, we put $C(z) = + \infty$ in the latter case since we do not want $Z$ to stop.} In this setting, we want to find a plan taking $Z$ from a state $s_{\mathrm{init}}$ to a state $s_{\mathrm{goal}}$ with minimal cumulative cost. In detail, let $U(s_{\mathrm{init}}, s_{\mathrm{goal}})$ be the set of all plans taking $Z$ from $s_{\mathrm{init}}$ to $s_{\mathrm{goal}}$, i.e., $u \in U(s_{\mathrm{init}}, s_{\mathrm{goal}})$ has an induced trajectory $z =(q_i,x_i)_{i=1}^N$ such that $q_1 =  s_{\mathrm{init}}$ and $q_{N+1} := \psi(q_N,x_N) = s_{\mathrm{goal}}$. Then we want to find an optimal plan $\hat{u} \in \argmin_{u \in U(s_{\mathrm{init}}, s_{\mathrm{goal}})} C(z)$ (where $z$ is the induced trajectory of $u$). We call $\hat{u}$ an optimal plan to the planning objective $(Z,s_{\mathrm{init}},s_{\mathrm{goal}})$ and the corresponding trajectory for an optimal trajectory. In Section \ref{control_algorithm_theory}, we present a control algorithm that efficiently computes optimal plans.

\subsection{Modifications}
In this section, we formalise changes in an HiMM, called modifications. More precisely, we introduce four kinds of modifications: state addition, state subtraction, arc modification and composition, specified by Definitions \ref{def:state_addition} to \ref{def:composition}. These modifications are schematically depicted in Fig.~\ref{fig:all_operations}.


\begin{definition}[State Addition]\label{def:state_addition}
Let $Z = (X,T)$ be an HiMM and $M \in X$ be an MM of $Z$. We say that an HiMM $Z_{\mathrm{add}} = (X_{\mathrm{add}},T_{\mathrm{add}})$ is added to $M$ in $Z$ forming a new HiMM $Z'$ with tree $T'$ equal to $T$ except that a new arc $M \xrightarrow{q} N$ has been added from $M$ to the root MM $N$ of $T_{\mathrm{add}}$ (connecting $Z_{\mathrm{add}}$ to $M$) labelled with some (arbitrary but distinct) state $q \notin Q(M)$, and $M$ has updated its state set to $Q \cup \{q\}$. Adding a state to $M$ in $Z$ is identical except that the arc $M \xrightarrow{q} N$ is instead a terminal arc $M \xrightarrow{q} \emptyset$.
\end{definition}

\begin{figure}
	\centering
  \includegraphics[width=0.45\textwidth]{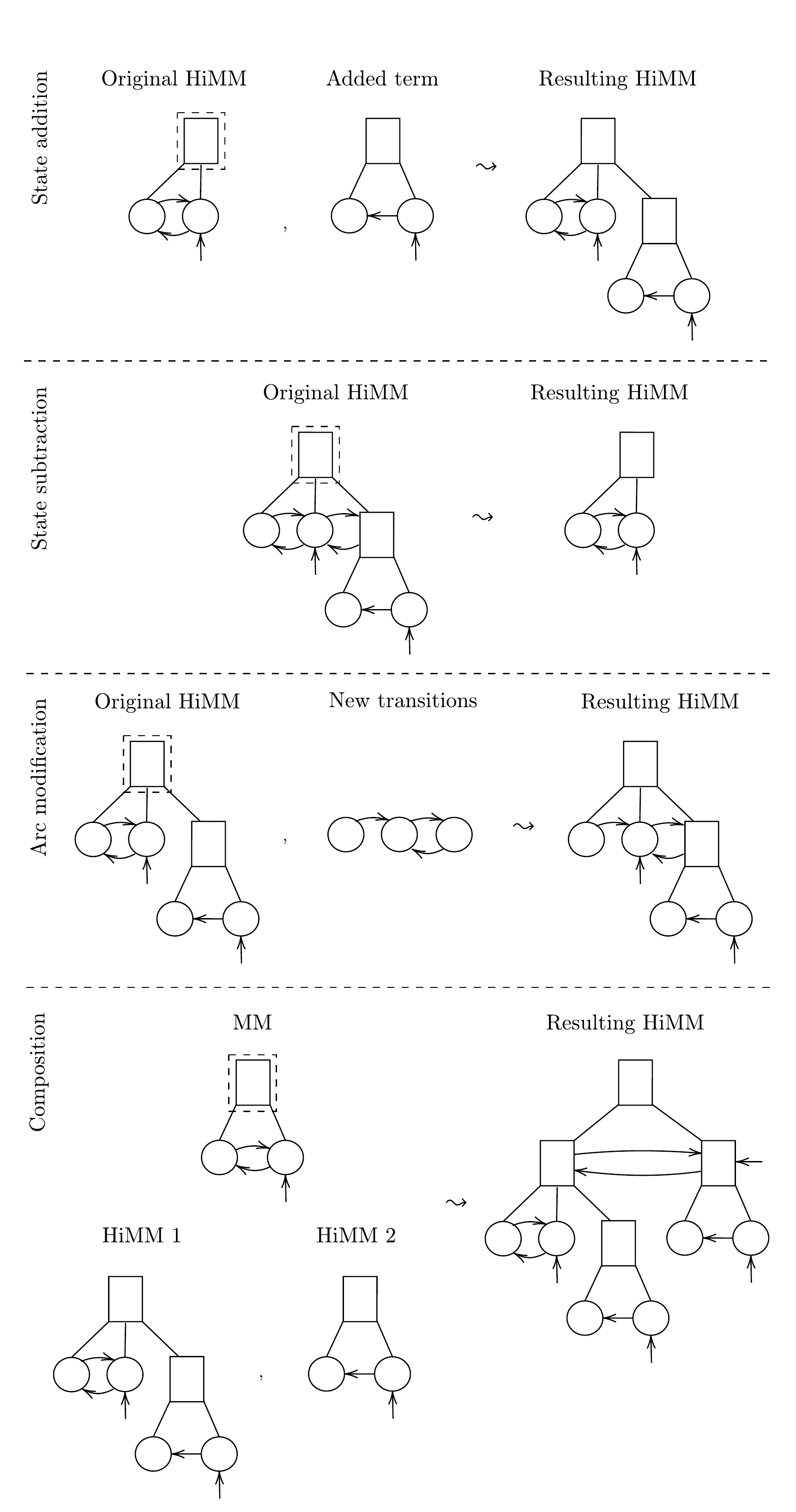}
  \caption{The four modifications of an HiMM depicted. The dashed rectangle specifies the MM $M$ subject to change.} 
  \label{fig:all_operations}
\end{figure}

\begin{definition}[State Subtraction]\label{def:state_subtraction}
Let $Z = (X,T)$ be an HiMM and $M \in X$ be an MM of $Z$. We say that we subtract a state $q \in Q(M)$ and form a new HiMM $Z'$ with tree $T'$ equal to $T$ except that the $q$-arc from $M$ is removed, and $M$ is the modified MM with state $q$ removed.\footnote{For simplicity, we forbid the start state to be removed. However, using an arc modification, one can change the start state and then remove the previous start state.}
\end{definition}

\begin{definition}[Arc Modification]\label{def:arc_modification}
Let $Z = (X,T)$ be an HiMM and $M = (Q,\Sigma,\Lambda,\delta,\gamma,s) \in X$ be an MM of $Z$. Let $\delta': Q \times \Sigma \rightharpoonup Q$, $\gamma': Q \times \Sigma \rightarrow \Lambda$ and $s' \in Q$ be the new transition function, output function and start state, and form $M' = (Q,\Sigma,\Lambda,\delta',\gamma',s')$. Replace $M$ with $M'$ in $Z$ with new HiMM $Z'$ as a result, called an arc~modification.
\end{definition}

\begin{definition}[Composition]\label{def:composition}
Let $M$ be an MM with states $Q = \{q_1,\dots,q_{|Q|}\}$ and $Z_{\mathrm{seq}} = \{ Z_1,\dots,Z_{n} \} $ be a set of HiMMs such that $n \leq |Q|$. The composition of $Z_{\mathrm{seq}}$ with respect $M$ is the HiMM $Z'$ with tree $T'$ having root $M$ and arcs $M \xrightarrow{q_i} R_i$ for $1 \leq i \leq n$, where $R_i$ is the root MM of $T_i$ (connecting $Z_i =(X_i,T_i)$ to $M$), and $M \xrightarrow{q_i} \emptyset$ for $i>n$.
\end{definition}

\begin{remark}[Intuition]
The composition replaces the first $n$ states in the MM $M$ with corresponding HiMMs from $Z_{\mathrm{seq}}$, while the remaining $|Q|-n$ states in $M$ are left unchanged.
\end{remark}

These four modifications are general enough to account for any change in the sense that any HiMM $Z=(X,T)$ can be changed to any other HiMM $Z'=(X',T')$ using these modifications. To see it, first add all the children of the root of $Z'$ (i.e., the root MM of $T'$) to the root of $Z$ (via repeated state addition), then do an arc modification so that the added part to the root of $Z$ is identical to the root of $Z'$ (including start state), and finally remove the original children of the root of $Z$ (using state subtraction). Then $Z$ is identical to $Z'$. Of course, this is a brute-force change, replacing the whole tree of $Z$ with the whole tree of $Z'$, and there may be other ways that use much fever modifications to change $Z$ into $Z'$. In particular, the true benefit of using modifications is instead when we have small changes in the system, such as removing just a subtree of $Z$ or changing the transitions in one of the MMs, captured by using just a few modifications.

\subsubsection{Reconfigurability}\label{reconfigurability_definition}
We say that a control algorithm is reconfigurable if a change in the HiMM given by any of the four modifications can be handled with low time complexity (where the change is assumed to be known to the algorithm). We propose such a control algorithm in Section~\ref{more_control_algorithm_theory}.








\subsection{Problem Statement}
We now formalise the problem statement. In this paper, we want to find a control algorithm that, for a given HiMM~$Z$:
\begin{enumerate}
\item Efficiently computes an optimal plan ${u}$ to any planning objective $(Z,s_{\mathrm{init}},s_{\mathrm{goal}})$, where efficiency means that the computation has low time complexity.
\item Is reconfigurable in the sense of Section \ref{reconfigurability_definition}.
\end{enumerate}
We present in Section \ref{control_algorithm_theory} a control algorithm from \cite{stefansson2023ecc} that fulfils the first criteria and extend this control algorithm in Section \ref{more_control_algorithm_theory} to fulfil the second criteria.

\section{Efficient Hierarchical Planning}\label{control_algorithm_theory}
In this section, we present the control algorithm from \cite{stefansson2023ecc} for efficiently computing an optimal plan for a given and static (i.e., none-changing) HiMM $Z$. Then, in Section \ref{more_control_algorithm_theory}, we extend this algorithm to the reconfigurable case.




\begin{algorithm}[t]
\caption{Static Hierarchical Planning}\label{alg:hierarchical_planning}
\begin{algorithmic}[1]
\Require HiMM$ \; Z =(X,T)$ and states $s_{\mathrm{init}}, s_{\mathrm{goal}}$.
\Ensure Optimal plan $u$ to $(Z,s_{\mathrm{init}}, s_{\mathrm{goal}})$.
\State \textbf{Optimal Exit Computer:}
\State $(c_x^M,z_x^M)_{x \in \Sigma}^{M \in X} \gets \mathrm{Compute\_optimal\_exits}(Z)$
\State \textbf{Optimal Planner:}
\State $z \gets \mathrm{Reduce\_and\_solve}(Z,s_{\mathrm{init}},s_{\mathrm{goal}},(c_x^M,z_x^M)_{x \in \Sigma}^{M \in X})$
\State $u \gets \mathrm{Expand}(z,(z_x^M)_{x \in \Sigma}^{M \in X}, Z)$
\end{algorithmic}
\end{algorithm}


We start with an overview of the algorithm, summarised by Algorithm \ref{alg:hierarchical_planning}. Given an HiMM $Z =(X,T)$, the algorithm consists of two steps. In the first step, the Optimal Exit Computer preprocesses the HiMM $Z$, computing optimal exit costs $(c_x^M)_{x \in \Sigma}$ for each MM $M \in X$, and corresponding trajectories $(z_x^M)_{x \in \Sigma}$ (line 2). This step needs to be done only once for the given $Z$. In the second step, the Online Planner considers any two states $s_{\mathrm{init}}$ and $s_{\mathrm{goal}}$, and computes an optimal plan to $(Z,s_{\mathrm{init}}, s_{\mathrm{goal}})$ by first reducing irrelevant parts of the tree of $Z$ (intuitively, parts of the tree of $Z$ not on the path between $s_{\mathrm{init}}$ and $s_{\mathrm{goal}}$) using the result from the Optimal Exit Computer, then obtains an optimal trajectory $z$ for the reduced HiMM (line 4), and finally expands $z$ to get an optimal plan $u$ to the original HiMM $Z$ (line 5).

We provide the details of the Optimal Exit Computer in Section \ref{offline_step} since it is crucial for the reconfigurable case in Section \ref{more_control_algorithm_theory}. However, the Online Planner is identical in the reconfigurable case. We therefore only highlight the efficiency of it in Section \ref{online_step}, refer to \cite{stefansson2023ecc} for details.

\subsection{Optimal Exit Computer}\label{offline_step}
In this section, we provide the details of the Optimal Exit Computer computing the optimal exit costs. To this end, we first need to define what we mean by optimal exit costs:

\begin{definition}[Optimal Exit Cost]\label{def:optimal_exit_costs} 
Consider HiMM~$Z=(X,T)$ and MM $M \in X$. We say that a state $q \in S_Z$ is contained in $M$ if $q_i$ is a descendant of $M$ in $T$ (e.g, 4 is descendant of $B$ in Fig. \ref{fig:Himm_transition_example_v2}), and an $(M,x)$-exit trajectory is a trajectory $z = (q_i,x_i)_{i=1}^N$ such that all $q_i$ are contained in $M$, while $\psi(q_N,x_N)$ is not, and $x_N=x$ (i.e., exits $M$ with $x$). To an $(M,x)$-exit trajectory $z$, let the $(M,x)$-exit cost be given by $\sum_{i=1}^{N-1} \chi(q_i,x_i)$ (we exclude the cost of $(q_N,x_N)$ since the transition is happening outside the subtree with root $M$), and let $c_x^M$ be the minimum over all $(M,x)$-exit costs. We call $c_x^M$ an optimal exit cost.
\end{definition}

We now provide the details of the Optimal Exit Computer, computing the optimal exit costs $(c_x^M)_{x \in \Sigma}$ for each MM $M \in X$ recursively over $T$, as follows. Let $M \in X$ be given with states $Q(M) = \{q_1,\dots,q_l\}$. For $q \in Q(M)$, define $c_x^{q}$ to be zero if $q \in S_Z$ is a state of $Z$, or $c_x^{q} = c_x^{M_{q}}$ otherwise, where $M_{q}$ is the MM corresponding to $q$. Intuitively, this is the extra cost one needs to exit $q$ with $x$ in $Z$, where, by recursion, we may assume that each $c_x^{q}$ is known. To compute $(c_x^M)_{x \in \Sigma}$, form the augmented MM $\hat{M}$ given by:
\begin{equation*}
\hat{M} := (\{q_1,\dots,q_l\} \cup \{E_x\}_{x \in \Sigma},\Sigma(M),\Lambda(M), \hat{\delta},\hat{\gamma},s(M)).
\end{equation*}
Intuitively, $\hat{M}$ is the same as $M$ except that whenever we are in a state $q_i$ of $M$ that does not support an input $x \in \Sigma$ (i.e., $\delta_M(q_i,x) = \emptyset$), then $\hat{M}$ transits to $E_x$ instead (hence the extra states $\{E_x\}_{x \in \Sigma}$ in $\hat{M}$). More precisely, the transition function $\hat{\delta}$ is given by
\begin{equation*}
\hat{\delta}(q_i,y) :=
\begin{cases}
\delta_M(q_i,y), & \delta_M(q_i,y) \neq \emptyset \\
E_y, & \mathrm{otherwise,}
\end{cases}
\end{equation*}
and $\hat{\delta}(E_x,y)$ is immaterial. Furthermore, the cost is
\begin{equation*}
\hat{\gamma}(q_i,y) := 
\begin{cases}
c_y^{q_i}+\gamma_M(q_i,y), & \delta_M(q_i,y) \neq \emptyset \\
c_y^{q_i}, & \mathrm{otherwise,}
\end{cases}
\end{equation*}
and $\hat{\gamma}(E_x,y)$ is immaterial. With this, we can search in $\hat{M}$ from $s(M)$ using Dijkstra's algorithm to find the minimal cumulative costs to all $(E_x)_{x \in \Sigma}$. Then, $c^M_x$ is equal to the minimal cumulative cost to $E_x$ \cite[Proposition 1]{stefansson2023ecc}. In Dijkstra's algorithm, we also get the corresponding optimal trajectories $(z_x^M)_{x \in \Sigma}$ to $(c_x^M)_{x \in \Sigma}$ for free. The algorithm is summarised by Algorithm \ref{alg:update_offline_step} (also given in \cite{stefansson2023ecc} as Algorithm 2) removing lines 3-5 and 16, with time complexity $O(N [b_s |\Sigma|+(b_s+|\Sigma|) \log(b_s+|\Sigma|)])$ \cite[Proposition 2]{stefansson2023ecc}. Here, $b_s$ is the maximum number of states in an MM of $Z$, and $N$ is the number of MMs of $Z$. In particular, assuming a bound on the number of states and inputs in an MM of $Z$, we get time complexity $O(N)$. For further details, see \cite{stefansson2023ecc}.

\subsection{Optimal Planner}\label{online_step}
In this section, we highlight the time complexity of the Online Planner, see \cite{stefansson2023ecc} for further details. More precisely, the time complexity of $\mathrm{Reduce\_and\_solve}$ (line 4 in Algorithm \ref{alg:hierarchical_planning}) is
$O \big (|\Sigma|^2 \mathrm{depth}(Z) + |\Sigma| \mathrm{depth}(Z) \log ( |\Sigma| \mathrm{depth}(Z) ) \big ) 
+
O \big ( [b_s |\Sigma | + b_s \log(b_s)] \mathrm{depth}(Z) \big ).$
In particular, assuming a bound on the number of inputs and states in an MM, we get time complexity $O(\mathrm{depth}(Z) \log ( \mathrm{depth}(Z)))$. Moreover, the time complexity for $\mathrm{Expand}$ (line 5 in Algorithm \ref{alg:hierarchical_planning}) is $O(\mathrm{depth}(Z) |u|)$ for obtaining the whole plan $u$ at once (where $|u|$ is the length of $u$), or $O(\mathrm{depth}(Z))$ for obtaining the next input in $u$ (ideal for sequential executions). In particular, the total time complexity for obtaining the next optimal input in $u$ is bounded by $O(\mathrm{depth}(Z) \log ( \mathrm{depth}(Z)))$. This should be compared with Dijkstra's algorithm, having a time complexity lower bounded by $O(V \log (V))$ \cite{DijkstraFibonacci}, where $V$ is the number of states in the HiMM, which in general could be exponential in $\mathrm{depth}(Z)$. The Online Planner therefore potentially saves huge computing time for large HiMMs. 

\section{Reconfigurable Hierarchical Planning}\label{more_control_algorithm_theory}
In this section, we extend Algorithm \ref{alg:hierarchical_planning} to be reconfigurable. To first see the need for this, note that the computational bottleneck of Algorithm \ref{alg:hierarchical_planning} typically comes from the Optimal Exit Computer, with time complexity $O(N)$, whereas the Optimal Planner computes the next optimal input in time $O(\mathrm{depth}(Z) \log ( \mathrm{depth}(Z)))$. Therefore, the Optimal Planner is typically much faster than the Optimal Exit Computer. This is not a problem for a non-changing HiMM $Z$, since the optimal exit costs need to be computed only once, which can then be used repeatedly by the Optimal Planner to rapidly compute optimal plans. However, if $Z$ changes, then one would need to re-compute all the optimal exit costs every time $Z$ is changed, which is inefficient, especially if the changes are small. 


\begin{figure}
	\centering
  \includegraphics[width=0.49\textwidth]{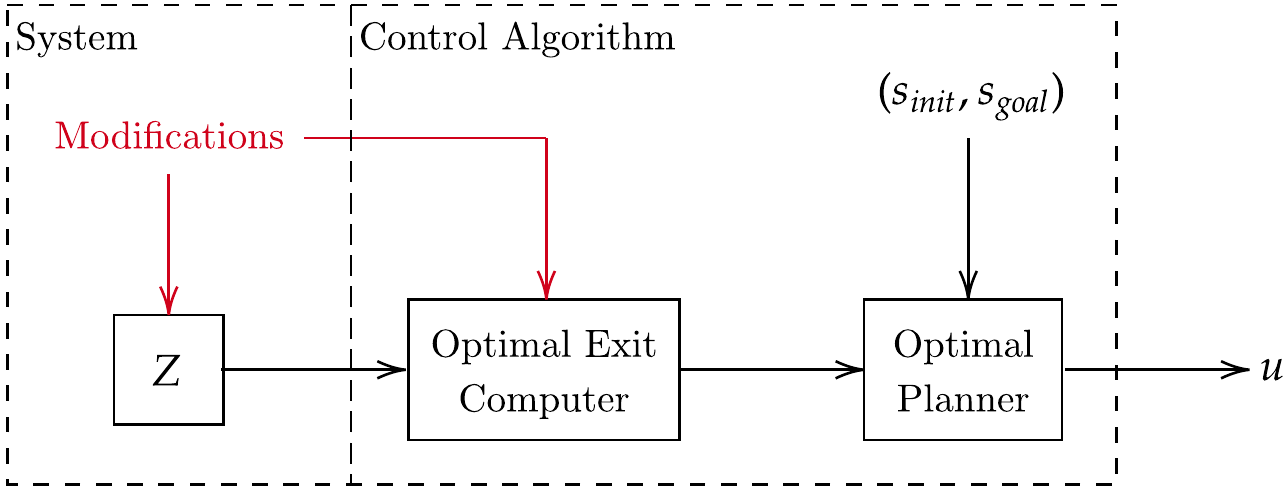}
  \caption{The control algorithm and its interaction with the changing system $Z$.} 
  \label{fig:algorithm_overview}
\end{figure}

To address this issue, we extend Algorithm \ref{alg:hierarchical_planning} so that the Optimal Exit Computer updates only the needed parts of its computation whenever $Z$ changes. The resulting control algorithm and its interaction with the (changing) HiMM $Z$ are depicted in Fig. \ref{fig:algorithm_overview}. Here, a change in $Z$ is given by a sequence of modifications. The control algorithm gets informed of the change and the modifications. By specifying the modifications, only some of the MMs of $Z$ need to be recomputed, marked by the Optimal Exit Computer. After all modifications have been reported, the Optimal Exit Computer recomputes the optimal exit costs, but only for the marked MMs. This potentially saves huge computation time. 




In more detail, the control algorithm is summarised by Algorithm \ref{alg:hierarchical_planning_new}, where changes with respect to Algorithm \ref{alg:hierarchical_planning} are in red. Here, $m_{\mathrm{seq}}$ is the sequence of modifications, read sequentially (line 2-4): For $m$ in $m_{\mathrm{seq}}$, $Z$ is first changed with respect to $m$ (line 3)\footnote{This line is executed by the system and not the control algorithm, see Fig. \ref{fig:algorithm_overview}. However, the algorithm tracks the change, so we keep it for clarity.}, then, the Optimal Exit Computer marks corresponding MMs (line 4), given by Algorithm \ref{alg:mark}. After $m_{\mathrm{seq}}$ has been read, the optimal exit costs for the marked MMs are updated (line 6), given by Algorithm \ref{alg:update_offline_step} (red parts accounts for the marking). We stress that we only need to execute the Optimal Exit Computer whenever $Z$ changes. The Optimal Planner is identical to the one in Algorithm \ref{alg:hierarchical_planning} and can be used repeatedly as long as changes are~{up-to-date}.

\subsection{Reconfigurability}\label{reconfigurability}

We now show that Algorithm \ref{alg:hierarchical_planning_new} is reconfigurable. We first explain the marking procedure in detail. Then, we prove the correctness of the Optimal Exit Computer (Proposition \ref{prop:update_offline_step_correctness}) and show that the time complexity is low for marking and updating the optimal exit costs for each of the four modifications (Proposition \ref{prop:update_offline_step_time_complexity}), thus, Algorithm \ref{alg:hierarchical_planning_new} is reconfigurable.

\begin{algorithm}[t]
\caption{Reconfigurable Hierarchical Planning}\label{alg:hierarchical_planning_new}
\begin{algorithmic}[1]
\Require HiMM$ \; Z$, \textcolor{BrickRed}{modifications $m_{seq}$} and states $s_{\mathrm{init}}, s_{\mathrm{goal}}$.
\Ensure Optimal plan $u$ to $(Z,s_{\mathrm{init}}, s_{\mathrm{goal}})$.
\State \textbf{Optimal Exit Computer:}
\For{\textcolor{BrickRed}{$m$ in $m_{seq}$}}
\State \textcolor{BrickRed}{$\mathrm{Modify}(Z,m)$} \Comment{\textcolor{BrickRed}{This line is executed in System}}
\State \textcolor{BrickRed}{$\mathrm{Mark}(Z,m)$}
\EndFor
\State $(c_x^M,z_x^M)_{x \in \Sigma}^{M \in X} \gets \mathrm{Compute\_optimal\_exits}(Z)$
\State \textbf{Optimal Planner:}
\State $z \gets \mathrm{Reduce\_and\_solve}(Z,s_{\mathrm{init}},s_{\mathrm{goal}},(c_x^M,z_x^M)_{x \in \Sigma}^{M \in X})$
\State $u \gets \mathrm{Expand}(z,(z_x^M)_{x \in \Sigma}^{M \in X}, Z)$
\end{algorithmic}
\end{algorithm}

\subsubsection{Marking}
The marking of $Z=(X,T)$ given a modification $m$ is done by Algorithm \ref{alg:mark}. In brief, the procedure is to mark the MM considered by $m$ and all ancestors to it.
In more detail, consider the case when $m$ is an arc modification modifying an MM $M$ in $Z$. With $M$ modified, the optimal exit costs of $M$ might have changed. We therefore mark $M$ (to recompute its optimal exit costs) in $Z$ (line 1). Also, modifying $M$ could affect the optimal exit costs of the MMs above $M$ in the hierarchy of $Z$, i.e., the ancestors of $M$. We therefore mark all the ancestors of $M$ (line 3). However, note that all other MMs in $Z$ have unaffected optimal exit costs, and we are therefore done. The case for state subtraction and state addition are analogous, and in the composition case, we only need to mark $M$ since $M$ has no~ancestors. 
\begin{convention}\label{marking_convention}
When initialising an HiMM $Z$, all MMs are by default marked since no optimal exit costs have been computed. This enables us to use Algorithm \ref{alg:update_offline_step} when initialising $Z$ and not only when updating $Z$ (after modifications). Also, whenever we modify an HiMM we always mark it accordingly (as above). Thus, we do not need to mark $Z_{\mathrm{add}}$ in state addition or $\{Z_1,\dots,Z_n\}$ in composition since their optimal exit costs do not change by the modification, and are already correctly marked (note that some of their optimal exit costs might have already been computed by Algorithm \ref{alg:update_offline_step}). See the proof of Proposition \ref{prop:update_offline_step_correctness} for details.
\end{convention}


\begin{algorithm}[t]
\caption{Mark}\label{alg:mark}
\begin{algorithmic}[1]
\Require Modification $m$ and (modified) HiMM $Z$
\Ensure Appropriate marking of $Z$
\State $\mathrm{mark\_itself}(M,Z)$
\If{$m$ not composition}
\State $\mathrm{mark\_ancestors}(M,Z)$
\EndIf
\end{algorithmic}
\end{algorithm}



\subsubsection{Computing optimal exit costs}  
After marking, the optimal exit costs for the marked MMs are computed using Algorithm \ref{alg:update_offline_step}, where the computation of the optimal exit costs for a given MM is as in Section \ref{offline_step}. The correctness of this procedure follows from the following proposition:



\begin{proposition}\label{prop:update_offline_step_correctness}
Let $Z$ be an HiMM, modified by the sequence of modifications $m_{\mathrm{seq}}$ (possibly empty, e.g., if $Z$ is only initialised) and marked using Algorithm \ref{alg:mark}. Then Algorithm \ref{alg:update_offline_step} returns the correct optimal exit costs. 
\end{proposition}

\subsubsection{Time complexity} Finally, we address the time complexity. The time complexity of Algorithm \ref{alg:mark} is $O(\mathrm{depth}(Z))$ since there are at most $\mathrm{depth}(Z)$ MMs to mark (the MM itself and all its ancestors). The time complexity of Algorithm \ref{alg:update_offline_step} varies with the modifications that have been done. In the worst case, all optimal exit costs need to be computed with time complexity $O(N [b_s |\Sigma|+(b_s+|\Sigma|) \log(b_s+|\Sigma|)])$ \cite[Proposition 2]{stefansson2023ecc}. However, for just one modification, we get:


\begin{proposition}\label{prop:update_offline_step_time_complexity}
Let HiMM $Z$ be given and consider the Optimal Exit Computer with only one modification $m_{\mathrm{seq}} = m$. Assume $Z$ have all optimal exit costs computed before the modification, and that the same holds for $Z_{\mathrm{add}}$ if $m$ is a state addition, or $\{Z_1,\dots,Z_n\}$ if $m$ is a composition. Then, the time complexity for executing the Optimal Exit Computer if $m$ is a state addition, state subtraction, or arc modification is $O( \mathrm{depth}(Z) [b_s |\Sigma|+(b_s+|\Sigma|) \log(b_s+|\Sigma|)] )$, and $O([b_s |\Sigma|+(b_s+|\Sigma|) \log(b_s+|\Sigma|)] )$ if $m$ is a composition. In particular, assuming a bound on the number of states and inputs of the MMs in $Z$, the time complexity for any modification is $O(\mathrm{depth}(Z))$. Thus, Algorithm \ref{alg:hierarchical_planning_new} is~reconfigurable.
\end{proposition}


\section{Numerical evaluations}\label{numerical_evaluations}
In this section, we present numerical evaluations, validating the proposed control algorithm given by Algorithm \ref{alg:hierarchical_planning_new}. We consider the robot application from the motivation and compare our method with Dijkstra's algorithm and Contraction Hierarchies, \if\longversion0 see \cite{stefansson2023cdc} for implementation details.\else
see Appendix for implementation~details.
\fi

\begin{algorithm}[t]
\caption{Compute\_optimal\_exits}\label{alg:update_offline_step}
\begin{algorithmic}[1]
\Require HiMM $Z = (X,T)$ with markings.
\Ensure Computed $(c_x^M,z_x^M)_{x \in \Sigma}$ for each MM $M$ of $Z$
\State $\mathrm{Optimal\_exit}(M_0)$ \Comment{Run from root MM $M_0$ in $T$}
\State $\mathrm{Optimal\_exit}(M)$: \Comment{Recursive help function} 
\If {\textcolor{BrickRed}{$M$ is not marked}}
\State return $(c_x^M,z_x^M)_{x \in \Sigma}$ \Comment{Have already been computed}
\Else
\For {each state $q$ in $Q(M)$}
\If{$q \in S_Z$}
\State $(c_x^q)_{x \in \Sigma} \gets 0_{|\Sigma|}$
\Else
\State Let $M_q$ be the MM corresponding to $q$.
\State $(c_x^{q},z_x^{q})_{x \in \Sigma} \gets \mathrm{Optimal\_exit}(M_q)$
\EndIf
\EndFor
\State Construct $\hat{M}$
\State $(c_x^M,z_x^M)_{x \in \Sigma}$ $\gets$ Dijkstra($s(M),\{E_x\}_{x \in \Sigma},\hat{M}$)
\State \textcolor{BrickRed}{Unmark $M$}
\State return $(c_x^M,z_x^M)_{x \in \Sigma}$
\EndIf
\end{algorithmic}
\end{algorithm}




\subsection{Case studies}
We first formalise the robot application from the motivation, also given in \cite{stefansson2023ecc}. The HiMM $Z$ is constructed using three MMs, corresponding to the layers in Fig. \ref{fig:robot_example_detailed_overview}.

The first MM (for Layer 1) has 10 states corresponding to the houses. The houses are ordered in a line where the robot can move to a neighbouring house as seen by the arrows in Fig. \ref{fig:robot_example_detailed_overview}, at a cost of 100. Start state is $s(M_1) = 1$.

The second MM (for Layer 2), $M_2$, corresponds to the locations in a house with $10\cdot 10+1=101$ states in a grid-like formation where the robot can move to any neighbouring grid-point, at a cost of 1, as in Fig. \ref{fig:robot_example_detailed_overview}, and start state $s(M_2)$ marked S, interpreted as the entrance to the house. From $s(M_2)$, the robot can also move left/right. However, these two inputs make the robot exit $M_2$ and instead, in $Z$, move it to the corresponding neighbouring house in $M_1$. 

The third MM (for Layer 3), $M_3$, corresponds to the lab desk at a location, with $10 \cdot 9+1 = 91$ states. Here the robot starts at $s(M_3)$, marked S in Fig. \ref{fig:robot_example_detailed_overview}. From $s(M_3)$, the robot can choose to go up, down, left or right, which quits $M_3$ and instead moves to the corresponding location in $M_2$; It can also start to steer a robot arm over a $3 \times 3$ test tube rack analogous to the location-grid in $M_2$, where the robot arm $(i,j)$ ($1 \leq i,j \leq 3$) is initially over tube (1,1). The robot can also scan a tube with the arm. The scanned tube $(i,j)$ is then remembered and no further tubes can be scanned. Finally, the robot can go back to $s(M_3)$ when the arm is again over tube (1,1). All costs in $M_3$ are set to 0.5, except scanning which costs 10. The hierarchy of $M_1$, $M_2$ and $M_3$ yields $Z$ with $10 \cdot 101 \cdot 91 = 91910$ states.

The case studies we consider are as follows. In Study 1, we consider $Z$ just described. In Study 2, we change $Z$ by adding a house, House 11, being a neighbour house to House 10 (see Fig. \ref{fig:robot_example_detailed_overview}), formally done by a state addition (adding House 11), followed by an arc modification (connecting House 10 and 11). All MMs of House 11 are marked, i.e., their optimal exit costs have not been computed yet. In Study 3, we also change $Z$ (from Study 1), but now by blocking some locations in House 2 as in Fig. \ref{fig:robot_example_change_overview}, formally removing locations using state subtraction. In all studies, the robot starts in the bottom-right corner of House 1 with an arm over tube $(2,2)$, with the goal to scan tube $(2,2)$ in the bottom-right corner of House 10, 11 and 2 in Study 1, 2 and 3 respectively.

\subsection{Result}

\begin{table}[]
\centering
\caption{Computing times for the case studies.}
\label{table_1}
\begin{tabular}{llll}
\hline
  &  Study 1 & Study 2  & Study 3  \\ \hline
Compute optimal plan           & 0.0254 s & 0.0271 s & 0.0228 s \\
Compute optimal exit costs            & 2.23 s & 0.208 s & 0.00234 s \\
Compute all optimal exit costs               & 2.23 s & 2.24 s & 1.96 s \\
Dijkstra's algorithm                & 3.73 s & 4.01 s & 1.37 s \\
CH – preprocessing               & 1196 s  & 1320 s & 1175 s \\
CH – compute optimal plan               & 0.0166 s & 0.0187 s & 0.0155 s \\
\hline  
\end{tabular}
\end{table}

The results of Study 1 to 3 are summarised in Table \ref{table_1}. The first row shows the time it takes the Optimal Planner in Algorithm \ref{alg:hierarchical_planning_new} to compute an optimal plan. In all three studies, the computation time is around 20-30 ms. This should be compared with Dijkstra's algorithm (row 4) with times around 1-4 s, two orders of magnitudes slower. Contraction Hierarchies (row 6) has instead a computation time around 16-19 ms, therefore slightly faster than Algorithm \ref{alg:hierarchical_planning_new}, though still the same order of magnitude. However, the slightly faster times come at the expense of a slower preprocessing step for the Contraction Hierarchies taking around 1200-1300 s to finish (row 5), compared with significantly faster computation times of Algorithm \ref{alg:update_offline_step} (row 2): 2.23 s, 0.208 s and 0.00234 s for Study 1, 2 and 3, respectively. Here, we also see the importance of reconfigurability, explaining the time differences of Algorithm \ref{alg:update_offline_step}. More precisely, in Study 1, we initialise $Z$ and need to compute all optimal exit costs. This takes the same time (2.23 s) as computing all optimal exit costs (row 3). However, in Study 2, we only modify $Z$ from Study 1, hence, only some optimal exit costs need to be recomputed, in this case, all MMs of $Z$ corresponding to House 11 (since they have not been computed yet) and the root MM of $Z$ (since it has been modified). The result is a much faster computation time (0.208 s), compared to computing all optimal exit costs (2.24 s). In Study 3, the modification results in even fewer MMs that need to be recomputed, with resulting time 2.34 ms. This shows the benefit of reconfigurability, being several orders of magnitude faster than needing to compute all optimal exit costs (as in \cite{stefansson2023ecc}).

\section{Conclusion}\label{conclusion}
In this paper, we have considered a planning problem for a large-scale system modelled as an HiMM $Z$ and developed a control algorithm for computing optimal plans between any two states. The control algorithm consists of two steps. A preprocessing step computing optimal exit costs for each MM in the $Z$, and a query step computing an optimal plan between any two states, using the optimal exit costs. Given bounds on the number of states and inputs of an MM in $Z$, the preprocessing step has time complexity $O(N)$ (where $N$ is the number of MMs in $Z$), while the query step has time complexity $O(\mathrm{depth}(Z) \log ( \mathrm{depth}(Z)))$ for obtaining the next optimal input. The latter enables rapid optimal plan computations. The control algorithm is also reconfigurable in the sense that changes, given by modifications, can be handled with ease, where a modification in $Z$ only needs a minor update of the preprocessing step in time $O(\mathrm{depth}(Z))$. We validated the control algorithm on a robotic application, comparing it with Dijkstra's algorithm and Contraction Hierarchies. We noted that our control algorithm outperforms Dijkstra's algorithm, computes optimal plans in the same order of magnitude as Contraction Hierarchies but with a significantly faster preprocessing step, and handles changes in the HiMM efficiently.

Future work includes extensions to stochastic hierarchical systems and additional numerical~evaluations. Another interesting direction is to modify the decomposition method in \cite{biggar2021modular} to automatically decompose an MM into an equivalent HiMM and then use our planning algorithm on the HiMM to efficiently compute optimal plans.

\bibliographystyle{plain}
\bibliography{Ref3}

\if\longversion1

\section*{Appendix}

\subsection{Proofs}


We start with the proof of Proposition \ref{prop:update_offline_step_correctness}, followed by the proof of Proposition \ref{prop:update_offline_step_time_complexity}.

\begin{proof}[Proof of Proposition \ref{prop:update_offline_step_correctness}]
Let HiMM $Z =(X,T)$ be given. The key observation for showing Proposition \ref{prop:update_offline_step_correctness} is to note that the marked MMs of $Z$ will always either form a subtree of $T$ that includes the root of $T$, or be empty (where the latter is the case if all MMs in $Z$ have optimal exit costs that are up-to-date). Therefore, Algorithm \ref{alg:update_offline_step} (starting at the root of $T$) will always reach all MMs that are marked and since the computation for a specific MM is correct (see Section \ref{control_algorithm_theory}), we conclude that Algorithm \ref{alg:update_offline_step} indeed returns the correct optimal exit costs and corresponding trajectories. Thus, it remains to show the key observation (i.e., that the marked MMs of $Z$ will always either form a subtree of $T$ that includes the root of $T$, or be empty), which is intuitively clear, but needs a technical proof to account for the HiMMs that one adds to form $Z$. The proof is given below, where we for brevity say that a HiMM $Z$ that fulfils the key observation has the subtree-property.

We prove the key observation by induction. More precisely, note first that $Z$ is formed by starting with a set $\mathcal{Z}$ of initialised HiMMs that through modifications eventually form $Z$. Let $n$ be the number of times we add HiMMs in this procedure using either state addition or composition. We prove that $Z$ has the subtree-property by induction over $n$. 

\subsubsection*{Base case}
Consider first the base case $n=0$. In this case, $\mathcal{Z}$ can only contain one element\footnote{Since if $\mathcal{Z}$ would contain at least two elements, then these HiMMs would eventually need to be composed or added to one another (to form $Z$ in the end), which is not possible if $n=0$.}, call it $Z_1$. Therefore, we only start with $Z_1$, which is first initialised and changed through modifications to eventually form $Z$, where we possibly also compute some of the optimal exit costs using Algorithm \ref{alg:update_offline_step} along the way (see Convention \ref{marking_convention}). Formally, this procedure is nothing but a sequence of changes on the form
\begin{equation}\label{eq:proof_sequence}
Z_1 \xrightarrow{o_1} Z^{(1)} \xrightarrow{o_2} Z^{(2)} \xrightarrow{o_3} \dots \xrightarrow{o_m} Z
\end{equation}
where $o_i$ is either a modification, or $o_i = \star$ specifying that we call Algorithm \ref{alg:update_offline_step}. More precisely, for a modification $o_i$, $Z \xrightarrow{o_i} Z'$ means that we changed $Z$ to $Z'$ through the modification $o_i$, while $Z \xrightarrow{\star} Z'$ means that we called Algorithm \ref{alg:update_offline_step} on $Z$ with $Z'$ as a result (in the latter case, note that $Z'$ only differs by which MMs are marked). We will prove that all HiMMs in \eqref{eq:proof_sequence} fulfil the subtree-property, by induction over the sequence \eqref{eq:proof_sequence}. For the base case, note that $Z_1$ trivially fulfils the subtree-property since all MMs are marked (due to initialisation, see Convention \ref{marking_convention}). For the induction step, assume that $W$ is an HiMM in \eqref{eq:proof_sequence} that fulfils the subtree-property, and changed to $W'$ through $o$. We prove that $W'$ then also fulfils the subtree-property, separating the proof into cases:
\begin{enumerate}[(i)]
\item Consider first the case when $o = \star$. Then we call Algorithm \ref{alg:update_offline_step} on $W$ to get $W'$. Since $W$ has the subtree-property, Algorithm \ref{alg:update_offline_step} will reach all marked MMs and unmark them. Hence, $W'$ has no marked MMs, and therefore fulfils the subtree-property.
\item Consider now the case when $o$ is a modification. If $o$ is a state subtraction or arc modification, on MM $M$ in $W$ say, then we mark $M$ and all the ancestors of $M$ in $W'$ (due to Convention \ref{marking_convention}).\footnote{In particular, we mark the root MM of $W'$.} Therefore, $W'$ fulfils the subtree-property. If $o$ is a state addition, then $o$ cannot add any HiMM $Z_{\mathrm{add}}$ to an MM $M$ (since this would contract $n=0$). Thus, $o$ can only add a state to an MM $M$. In this case, $M$ and all ancestors of are marked, so $W'$ again fulfils the subtree-property. Finally, $o$ cannot be a composition since $n=0$.
\end{enumerate}
Combining the two cases, we conclude that $W'$ fulfils the subtree-property. By induction over the sequence \eqref{eq:proof_sequence}, all HiMMs in \eqref{eq:proof_sequence} fulfil the subtree-property. In particular, $Z$ fulfils the subtree-property. We have therefore proved the base case $n=0$.

\subsubsection*{Induction step}
We continue with the induction step. Towards this, assume $n\geq1$. Let $Z_1$ be the last HiMM that is formed by either adding a HiMM using state addition, or using composition.\footnote{Note that $Z_1$ must indeed exist since we start with the HiMMs in $\mathcal{Z}$ but eventually are left with just $Z$.} Thus, from $Z_1$, $Z$ is formed by a sequence of changes as in \eqref{eq:proof_sequence}.\footnote{This follows from the fact that we cannot add another HiMM from $Z_1$ to $Z$, since $Z_1$ was the last HiMM that was formed by adding a HiMM.} Therefore, by just following the proof for the case $n=0$, we know that $Z$ fulfils the subtree-property if $Z_1$ fulfils the subtree-property. It remains to prove that $Z_1$ fulfils the subtree-property. By assumption, $Z_1$ is formed by adding a HiMM using state addition or composition. We consider the two cases separately:
\begin{enumerate}
\item Assume first that $Z_1$ is formed by state addition, adding an HiMM $Z_{\mathrm{add}}$ to an HiMM $Z_{\mathrm{to}}$. Let $n_{\mathrm{add}}$ ($n_{\mathrm{to}}$) be the number of times we add HiMMs to form $Z_{\mathrm{add}}$ ($Z_{\mathrm{to}}$) using either state addition or composition, similar to $n$. Then, $n_{\mathrm{add}} < n$ and $n_{\mathrm{to}} < n$.\footnote{This is because every time we add an HiMM to form e.g., $Z_{\mathrm{add}}$, this also counts as adding an HiMM to form $Z$, and we add at least one more HiMM to form $Z$, due to the formation of $Z_1$. Hence, $n_{\mathrm{add}} < n$. The argument for $n_{\mathrm{to}} < n$ is analogous.} By induction, $Z_{\mathrm{add}}$ and $Z_{\mathrm{to}}$ fulfils the subtree-property. Let $Z_{\mathrm{add}}$ be added to the MM $M$ in $Z_{\mathrm{to}}$. Then $M$ and its ancestors are marked\footnote{Note that this connects the marked subtree of $Z_{\mathrm{add}}$ with the marked subtree of $Z_{\mathrm{to}}$.} in $Z_1$ and since both $Z_{\mathrm{add}}$ and $Z_{\mathrm{to}}$ fulfils the subtree property, we conclude that $Z_1$ fulfils the~subtree-property.
\item Assume now that $Z_1$ is formed by composition using an MM $M$. By induction, analogous to the state addition case, all the HiMMs in the composition fulfils the subtree-property, and since $M$ (the root of $Z_1$) is marked, we conclude that $Z_1$ also fulfils the~subtree-property.
\end{enumerate}
Combining the two cases, we conclude that $Z_1$ fulfils the subtree-property, hence, $Z$ fulfils the subtree-property. 

We have showed both the base case and the induction step. Therefore, by induction over $n$, $Z$ always fulfils the subtree-property, proving the key observation. Therefore, the proof of Proposition \ref{prop:update_offline_step_correctness} is~complete.
\end{proof}

We continue with the proof of Proposition \ref{prop:update_offline_step_time_complexity}.

\begin{proof}[Proof of Proposition \ref{prop:update_offline_step_time_complexity}]
Let HiMM $Z=(X,T)$ be given. Following the proof of Proposition 2 in \cite{stefansson2023ecc}, it is easy to conclude that if $K$ MMs are marked in $Z$, then the time complexity of Algorithm \ref{alg:update_offline_step} is $O(K [b_s |\Sigma|+(b_s+|\Sigma|) \log(b_s+|\Sigma|)])$. Since $Z$ has no marked MMs before the modification $m$, the number of MMs in $Z$ after the modification is at most $\mathrm{depth}(Z)$ if $m$ is a state addition, state subtraction, or arc modification, and at most 1 if $m$ is a composition. Hence, marking takes time $O(\mathrm{depth}(Z))$ for state addition, state subtraction, or arc modification, and $O(1)$ composition. Furthermore, using the observation above, running Algorithm \ref{alg:update_offline_step} therefore takes time $O(\mathrm{depth}(Z) [b_s |\Sigma|+(b_s+|\Sigma|) \log(b_s+|\Sigma|)])$ for state addition, state subtraction, or arc modification, and $O( 1 [b_s |\Sigma|+(b_s+|\Sigma|) \log(b_s+|\Sigma|)])$ for composition. We conclude that the time complexity for executing the Optimal Exit Computer if $m$ is a state addition, state subtraction, or arc modification is $O( \mathrm{depth}(Z) [b_s |\Sigma|+(b_s+|\Sigma|) \log(b_s+|\Sigma|)] )$, and $O([b_s |\Sigma|+(b_s+|\Sigma|) \log(b_s+|\Sigma|)] )$ if $m$ is a composition.
\end{proof}

\subsection{Implementation Details}
In this section, we provide additional details concerning the implementation used in the numerical investigations in Section \ref{numerical_evaluations}, with focus on Contraction Hierarchies. First of all, all implementations are in Python. Moreover, all three methods (Algorithm \ref{alg:hierarchical_planning_new}, Dijkstra's algorithm and Contraction Hierarchies) uses Dijkstra's algorithm in their implementation. Theoretically, Dijkstra's algorithm has the lowest time complexity when using Fibonacci heaps \cite{DijkstraFibonacci}. However, in practice, better performance is typically achieved (at least in our case) with a simple priority queue, hence, in the simulations, we use the latter.

We continue with the details concerning our implementation of Contraction Hierarchies. First of all, Contraction Hierarchies is a shortest-path algorithm for graphs involving a preprocessing step and a query step. In the preprocessing step, one sequentially removes one node at a time in the graph (called a contraction), adding additional arcs (called shortcuts) to account for the removed node (the shortcuts are added so that the remaining nodes have the same shortest-distance to each other). This contraction procedure is then iterated until there is no nodes left to contract. Then, in the query step, an optimal plan is computed using a bidirectional search on the original graph but with the added shortcuts from the preprocessing step, with intuition that the added shortcuts improve the search speed. See \cite{geisberger2012exact} for details.

To get a good performance using Contraction Hierarchies, the node order selection in the preprocessing step is crucial (i.e., how we select the next node to contract). In our implementation, we picked three common heuristics for selecting the next node. Namely, we conduct lazy evaluations, where we pick the next node $n$ based on the edge difference of $n$, and add this number with the number of neighbours to $n$ that has already been contracted (to account for sparsity). The details of these heuristics are beyond the scope of this paper; instead we refer to the lecture notes given by Hannah Bast at \url{https://ad-wiki.informatik.uni-freiburg.de/teaching/EfficientRoutePlanningSS2012} (which our implementation is based on) and the paper \cite{geisberger2012exact}. To speed-up the preprocessing step (following the lecture notes), we also set a maximum number of visited nodes when Dijkstra's algorithm is used in a contraction to locally search if shortcuts should be added, setting the maximum to 20 nodes. Again, see the lecture notes and the paper \cite{geisberger2012exact} for details. Finally, we stress that there might exists other variations of Contraction Hierarchies that might perform better than our implementation (e.g., using other heuristics).




\fi

\end{document}